\documentclass{article}
\usepackage{amsfonts}
\usepackage{amsmath}

\newtheorem{theorem}{Theorem}

\newtheorem{corollary}[theorem]{Corollary}

\newtheorem{definition}[theorem]{Definition}
\newtheorem{example}[theorem]{Example}

\newtheorem{lemma}[theorem]{Lemma}

\newtheorem{proposition}[theorem]{Proposition}

\newenvironment{proof}[1][Proof]{\noindent\textbf{#1.} }{\ \rule{0.5em}{0.5em}}
\input{tcilatex}

\begin{document}

\title{On cyclic DNA codes over the rings $Z_{4}+wZ_{4}$ and $%
Z_{4}+wZ_{4}+vZ_{4}+wvZ_{4}$}
\author{Abdullah Dertli$^{a}$, Yasemin Cengellenmis$^{b}$ \and $\left(
a\right) $ Ondokuz May\i s University, Faculty of Arts and Sciences, \and %
Mathematics Department, Samsun, Turkey \and abdullah.dertli@gmail.com \and $%
\left( b\right) $ Trakya University, Faculty of Arts and Sciences, \and %
Mathematics Department, Edirne, Turkey \and ycengellenmis@gmail.com}
\maketitle

\begin{abstract}
The structures of cyclic DNA codes of odd length over the finite rings $%
R=Z_{4}+wZ_{4}$, $w^{2}=2$ and $%
S=Z_{4}+wZ_{4}+vZ_{4}+wvZ_{4},w^{2}=2,v^{2}=v,wv=vw$ are studied. The links
between the elements of the rings $R$, $S$ and 16 and 256 codons are
established, respectively. Cyclic codes of odd length over the finite ring $%
R $ satisfies reverse complement constraint and cyclic codes of odd length
over the finite ring $S$ satisfy reverse constraint and reverse complement
constraint\ are studied. Binary images of the cyclic DNA codes over the
finite rings $R$ and $S$ are determined. Moreover, a family of DNA skew
cyclic codes over $R$ is constructed, its property of being reverse
complement is studied.
\end{abstract}

\section{Introduction}

DNA is formed by the strands and each strands is sequence consists of four
nucleotides ; adenine (A), guanine (G), thymine (T) and cytosine (C). Two
strands of DNA are linked with Watson-Crick Complement. This is as $%
\overline{A}=T$, $\overline{T}=A$, $\overline{G}=C$, $\overline{C}=G$. For
example if $c=\left( ATCCG\right) $ then its complement is $\overline{c}%
=\left( TAGGC\right) $.

A code is called DNA codes if it satisfies some or all of the following
conditions;

\ \ \ \ i) \ The Hamming contraint, for any two different codewords $%
c_{1},c_{2}\in C,$ $H(c_{1},c_{2})\geq d$

\ \ \ ii) \ The reverse constraint, for any two different codewords $%
c_{1},c_{2}\in C,$ $H(c_{1},c_{2}^{r})\geq d$

\ \ \ iii) The reverse complement constraint, for any two different
codewords $c_{1},c_{2}\in C,$ $H(c_{1},c_{2}^{rc})\geq d$

\ \ \ iv) The fixed GC content constraint, for any codeword $c\in C$
contains the some number of G and C element.

The purpose of the i-iii contraints is to avoid undesirable hybridization
between different strands.

DNA computing were started by Leonhard Adleman in 1994, in [3]. The special
error correcting codes over some finite fields and finite rings with $4^{n}$
elements where $n\in N$ were used for DNA computing applications.

In [12], the reversible codes over finite fields were studied, firstly. It
was shown that $C=\left\langle f(x)\right\rangle $ is reversible if and only
if $f(x)$ is a self reciprocal polynomial. In [1], they developed the theory
for constructing linear and additive cyclic codes of odd length over $GF(4)$%
. In [13], they introduced a new family of polynomials which generates
reversible codes over a finite field $GF(16).$

In [2], the reversible cyclic codes of any length $n$ over the ring $Z_{4}$
were studied. A set of generators for cyclic codes over $Z_{4}$ with no
restrictions on the length $n$ was found. In [17], cyclic DNA codes over the
ring $R=\{0,1,u,1+u\}$ where $u^{2}=1$ based on a similarity measure were
constructed. In [9], the codes over the ring $F_{2}+uF_{2},u^{2}=0$ were
constructed for using in DNA computing applications.

I. Siap et al. considered cyclic DNA codes over the finite ring $%
F_{2}[u]/\left\langle u^{2}-1\right\rangle $ in [18]. In [10], Liang and
Wang considered cyclic DNA codes over $F_{2}+uF_{2},u^{2}=0.$ Y\i ld\i z and
Siap studied cyclic DNA codes over $F_{2}[u]/\left\langle
u^{4}-1\right\rangle $ in [19]. Bayram et al. considered codes over the
finite ring $F_{4}+vF_{4},v^{2}=v$ in [3]. Zhu and Chan studied cyclic DNA\
codes over the non-chain ring $F_{2}[u,v]/\left\langle
u^{2},v^{2}-v,uv-vu\right\rangle $ in [20]. In [5], Bemenni at al. studied
cyclic DNA codes over $F_{2}[u]/\left\langle u^{6}\right\rangle $.
Pattanayak et al. considered cyclic DNA codes over the ring $%
F_{2}[u,v]/<u^{2}-1,v^{3}-v,uv-vu>$ in [15]. Pattanayak and Singh studied
cyclic DNA codes over the ring $Z_{4}+uZ_{4},u^{2}=0$ in [14].

J. Gao et al. studied the construction of cyclic DNA codes by cyclic codes
over the finite ring $F_{4}[u]/\left\langle u^{2}+1\right\rangle $, in [11].
Also, the construction of DNA cyclic codes have been discussed by several
authors in [7,8,16].

We study families of DNA cyclic codes of the finite rings $Z_{4}+wZ_{4}$, $%
w^{2}=2$ and $Z_{4}+wZ_{4}+vZ_{4}+wvZ_{4},w^{2}=2,v^{2}=v,wv=vw.$ The rest
of the paper is organized as follows. In section 2, details about algebraic
structure of the finite ring $Z_{4}+wZ_{4}$, $w^{2}=2$ are given. We define
a Gray map from $R$ to $Z_{4}$. In section 3, cyclic codes of odd length
over $R$ satisfies the reverse complement constraint are determined. In
section 4, cyclic codes of odd length over $S$ satisfy the reverse
complement constraint and the reverse contraint are examined. A linear code
over $S$ is represented by means of two linear codes over $R$. In section 5,
the binary image of cyclic DNA code over $R$ is determined. In section 6,
the binary image of cyclic DNA code over $S$ is determined. In section 7, by
using a non trivial automorphism, the DNA skew cyclic codes are introduced.
In section 8, the design of linear DNA code is presented.

\section{Preliminaries}

The algebraic structure of the finite ring $R=Z_{4}+wZ_{4}$, $w^{2}=2$ is
given\ in [6]. $R$ is the commutative, characteristic 4 ring $%
Z_{4}+wZ_{4}=\{a+wb:a,b\in Z_{4}\}$ with $w^{2}=2$. $R$ can also be thought
of as the quotient ring $Z_{4}[w]/\left\langle w^{2}-2\right\rangle $. $R$
is principal ideal ring with 16 elements and finite chain ring. The units of
the ring are%
\begin{equation*}
1,3,1+w,3+w,1+2w,1+3w,3+3w,3+2w
\end{equation*}

and the non units are 
\begin{equation*}
0,2,w,2w,3w,2+w,2+2w,2+3w
\end{equation*}

$R$ has 4 ideals in all

\begin{eqnarray*}
\left\langle 0\right\rangle &=&\{0\} \\
\left\langle 1\right\rangle &=&\left\langle 3\right\rangle =\left\langle
1+3w\right\rangle =...=R \\
\left\langle w\right\rangle &=&\{0,2,w,2w,3w,2+w,2+2w,2+3w\} \\
&=&\left\langle 3w\right\rangle =\left\langle 2+w\right\rangle =\left\langle
2+3w\right\rangle \\
\left\langle 2w\right\rangle &=&\{0,w\} \\
\left\langle 2\right\rangle &=&\left\langle 2+2w\right\rangle
=\{0,2,2w,2+2w\}
\end{eqnarray*}

\begin{equation*}
\left\langle 0\right\rangle \subset \left\langle 2w\right\rangle \subset
\left\langle 2\right\rangle \subset \left\langle w\right\rangle \subset R
\end{equation*}%
Moreover $R$ is Frobenious ring.

We define%
\begin{eqnarray*}
\phi &:&R\longrightarrow Z_{4}^{2} \\
\phi \left( a+wb\right) &=&\left( a,b\right)
\end{eqnarray*}%
This Gray map is extended component wise to 
\begin{eqnarray*}
\phi &:&R^{n}\longrightarrow Z_{4}^{2n} \\
\left( \alpha _{1},\alpha _{2},...,\alpha _{n}\right)
&=&(a_{1},...,a_{n},b_{1},...,b_{n})
\end{eqnarray*}%
where $\alpha _{i}=a_{i}+b_{i}w$ with $i=1,2,...,n.$ $\phi $ is $Z_{4}$
module isomorphism.

A linear code $C$ of length $n$ over $R$ is a $R$-submodule of $R^{n}.$ An
element of $C$ is called a codeword. A code of length $n$ is cyclic if the
code invariant under the automorphism $\sigma $ which 
\begin{equation*}
\sigma \left( c_{0},c_{1},...,c_{n-1}\right) =\left(
c_{n-1},c_{0},...,c_{n-2}\right)
\end{equation*}

A cyclic code of length $n$ over $R$ can be identified with an ideal in the
quotient ring $R[x]/\left\langle x^{n}-1\right\rangle $ via the $R$--modul
isomorphism%
\begin{eqnarray*}
R^{n} &\longrightarrow &R[x]/\left\langle x^{n}-1\right\rangle \\
\left( c_{0},c_{1},...,c_{n-1}\right) &\longmapsto
&c_{0}+c_{1}x+...+c_{n-1}x^{n-1}+\left\langle x^{n}-1\right\rangle
\end{eqnarray*}

\begin{theorem}
Let $C$ be a cyclic code in $R[x]/\left\langle x^{n}-1\right\rangle .$Then
there exists polynomials $g(x),a(x)$ such that $a(x)|g(x)|x^{n}-1$ and $%
C=\left\langle g(x),wa(x)\right\rangle .$

The ring $R[x]/\left\langle x^{n}-1\right\rangle $ is a principal ideal ring
when $n$ is odd. So, if $n$ is odd, then there exist $s(x)\in
R[x]/\left\langle x^{n}-1\right\rangle $ such that $C=\left\langle
s(x)\right\rangle .$
\end{theorem}

\section{The reversible complement codes over R}

In this section, we study cyclic code of odd length over $R$ satisfies the
reverse complement constraint. Let $\{A,T,G,C\}$ represent the DNA alphabet.
DNA occurs in sequences with represented by sequences of the DNA alphabet.
DNA\ code of length $n$ is defined as a set of the codewords $\left(
x_{0},x_{1},...,x_{n-1}\right) $ where $x_{i}\in \{A,T,G,C\}.$ These
codewords must satisfy the four constraints which are mentioned in [20].

Since the ring $R$ is of cardinality 16, we define the map $\phi $ which
gives a one to one correspondence between the elements of $R$ and the 16
codons over the alphabet $\{A,T,G,C\}^{2}$ by using the Gray map as follows%
\begin{equation*}
\begin{tabular}{ccc}
Elements & Gray images & DNA double pairs \\ 
$0$ & $(0,0)$ & $AA$ \\ 
$1$ & $(1,0)$ & $CA$ \\ 
$2$ & $(2,0)$ & $GA$ \\ 
$3$ & $(3,0)$ & $TA$ \\ 
$w$ & $(0,1)$ & $AC$ \\ 
$2w$ & $(0,2)$ & $AG$ \\ 
$3w$ & $(0,3)$ & $AT$ \\ 
$1+w$ & $(1,1)$ & $CC$ \\ 
$1+2w$ & $(1,2)$ & $CG$ \\ 
$1+3w$ & $(1,3)$ & $CT$ \\ 
$2+w$ & $(2,1)$ & $GC$ \\ 
$2+2w$ & $(2,2)$ & $GG$ \\ 
$2+3w$ & $(2,3)$ & $GT$ \\ 
$3+w$ & $(3,1)$ & $TC$ \\ 
$3+2w$ & $(3,2)$ & $TG$ \\ 
$3+3w$ & $(3,3)$ & $TT$%
\end{tabular}%
\end{equation*}

The codons satisfy the Watson-Crick Complement.

\begin{definition}
For $x=\left( x_{0},x_{1},...,x_{n-1}\right) \in R^{n},$ the vector \ $%
\left( x_{n-1},x_{n-2},...,x_{1},x_{0}\right) $ is called the reverse of $x$
and is denoted by $x^{r}.$ A linear code $C$ of length $n$ over $R$ \ is
said to be reversible if $x^{r}\in C$ for \ every $x\in C.$

For $x=\left( x_{0},x_{1},...,x_{n-1}\right) \in R^{n},$ the vector \ $%
\left( \overline{x}_{0},\overline{x}_{1},...,\overline{x}_{n-1}\right) $ is
called the complement of $x$ and is denoted by $x^{c}.$ A linear code $C$ of
length $n$ over $R$ \ is said to be complement if $x^{c}\in C$ for \ every $%
x\in C.$

For $x=\left( x_{0},x_{1},...,x_{n-1}\right) \in R^{n},$ the vector \ $%
\left( \overline{x}_{n-1},\overline{x}_{n-2},...,\overline{x}_{1},\overline{x%
}_{0}\right) $ is called the reversible complement of $x$ and is denoted by $%
x^{rc}.$ A linear code $C$ of length $n$ over $R$ \ is said to be reversible
complement if $x^{rc}\in C$ for \ every $x\in C.$
\end{definition}

\begin{definition}
Let $f(x)=a_{0}+a_{1}x+...+a_{r}x^{r}$ with $a_{r}\neq 0$ be polynomial. The
reciprocal of $f(x)$ is defined as $f^{\ast }(x)=x^{r}f(\frac{1}{x}).$ It is
easy to see that $\deg f^{\ast }(x)\leq \deg f(x)$ and if $a_{0}\neq 0$,
then $\deg f^{\ast }(x)=\deg f(x).$ $f(x)$ is called a self reciprocal
polynomial if there is a constant $m$ such that $f^{\ast }(x)=mf(x).$
\end{definition}

\begin{lemma}
Let $f(x),g(x)$ be polynomials in $R[x].$ Suppose $\deg f(x)-\deg g(x)=m$
then,
\end{lemma}

\ \ \ \ \ \ \ i) $(f(x)g(x))^{\ast }=f^{\ast }(x)g^{\ast }(x)$

\ \ \ \ \ \ \ ii) $(f(x)+g(x))^{\ast }=f^{\ast }(x)+x^{m}g^{\ast }(x)$

\begin{lemma}
For any $a\in R,$ we have $a+\overline{a}=3+3w.$
\end{lemma}

\begin{lemma}
If $a\in \{0,1,2,3\}$, then we have $(3+3w)-\overline{wa}=wa.$
\end{lemma}

\begin{theorem}
Let $C=\left\langle g(x),wa(x)\right\rangle $ be a cyclic code of odd length 
$n$ over $R.$ If $f(x)^{rc}\in C$ for any $f(x)\in C$, then $%
(1+w)(1+x+x^{2}+...+x^{n-1})\in C$ and there are two constants $e,d\in
Z_{4}^{\ast }$ such that $g^{\ast }(x)=eg(x)$ and $a^{\ast }(x)=da(x).$
\end{theorem}

\begin{proof}
Suppose that $C=\left\langle g(x),wa(x)\right\rangle ,$ where $%
a(x)|g(x)|x^{n}-1\in Z_{4}[x].$ Since $(0,0,...,0)\in C$, then its
reversible complement is also in $C$. 
\begin{eqnarray*}
(0,0,...,0)^{rc} &=&(3+3w,3+3w,...,3+3w) \\
&=&3(1+w)(1,1,...,1)\in C
\end{eqnarray*}%
This vector corresponds of the polynomial%
\begin{equation*}
(3+3w)+(3+3w)x+...+(3+3w)x^{n-1}=(3+3w)\frac{x^{n}-1}{x-1}\in C
\end{equation*}%
Since $3\in Z_{4}^{\ast }$, then $(1+w)(1+x+...+x^{n-1})\in C.$

Let $g(x)=g_{0}+g_{1}x+...+g_{r-1}x^{r-1}+g_{r}x^{r}$. Note that $%
g(x)^{rc}=(3+3w)+(3+3w)x+...+(3+3w)x^{n-r-2}+\overline{g}_{r}x^{n-r-1}+...+%
\overline{g}_{1}x^{n-2}+\overline{g}_{0}x^{n-1}\in C.$

Since $C$ is a linear code, then%
\begin{equation*}
3(1+w)(1+x+x^{2}+...+x^{n-1})-g(x)^{rc}\in C
\end{equation*}%
which implies that $((3+3w)-\overline{g}_{r})x^{n-r-1}+((3+3w)-\overline{g}%
_{r-1})x^{n-r-2}+...+((3+3w)-\overline{g}_{0})x^{n-1}\in C$. By using $%
(3+3w)-\overline{a}=a$, this implies that%
\begin{equation*}
x^{n-r-1}(g_{r}+g_{r-1}x+...+g_{0}x^{r})=x^{n-r-1}g^{\ast }(x)\in C
\end{equation*}%
Since $g^{\ast }(x)\in C,$ this implies that 
\begin{equation*}
g^{\ast }(x)=g(x)u(x)+wa(x)v(x)
\end{equation*}%
where $u(x),v(x)\in Z_{4}[x].$ Since $g_{i}\in Z_{4}$ for $i=0,1,...,r$, \
we have that $v(x)=0$. As $\deg g^{\ast }(x)=\deg g(x),$ we have $u(x)\in
Z_{4}^{\ast }$. Therefore there is a constant $e\in Z_{4}^{\ast }$ such that 
$g^{\ast }(x)=eg(x)$. So, $g(x)$ is a self reciprocal polynomial.

Let $a(x)=a_{0}+a_{1}x+...+a_{t}x^{t}$. Suppose that $%
wa(x)=wa_{0}+wa_{1}x+...+wa_{t}x^{t}$. Then 
\begin{equation*}
(wa(x))^{rc}=(3+3w)+(3+3w)x+...+\overline{wa_{t}}x^{n-t-1}+...+\overline{%
wa_{1}}x^{n-2}+\overline{wa_{0}}x^{n-1}\in C
\end{equation*}%
As $(3+3w)\frac{x^{n}-1}{x-1}\in C$ and $C$ is a linear code, then%
\begin{equation*}
-(wa(x))^{rc}+(3+3w)\frac{x^{n}-1}{x-1}\in C
\end{equation*}%
Hence, $x^{n-t-1}[(-(\overline{wa_{t}})+(3+3w))+(-(\overline{wa_{t-1}}%
)+(3+3w))x+...+(-(\overline{wa_{0}})+(3+3w))x^{t}]$. By Lemma 6, we get 
\begin{equation*}
x^{n-t-1}(wa_{t}+wa_{t-1}x+...+wa_{0}x^{t})
\end{equation*}%
$x^{n-t-1}wa^{\ast }(x)\in C$. Since $wa^{\ast }(x)\in C$, we have 
\begin{equation*}
wa^{\ast }(x)=g(x)h(x)+wa(x)s(x)
\end{equation*}

Since $w$ doesn't appear in $g(x)$, it follows that $h(x)=0$ and $a^{\ast
}(x)=a(x)s(x)$. As $\deg a^{\ast }(x)=\deg a(x)$, then $s(x)\in Z_{4}^{\ast
} $. So, $a(x)$ is a self reciprocal polynomial.
\end{proof}

\begin{theorem}
Let $C=\left\langle g(x),wa(x)\right\rangle $ be a cyclic code of odd length 
$n$ over $R.$ If $(1+w)(1+x+x^{2}+...+x^{n-1})\in C$ and $g(x),a(x)$ are
self reciprocal polynomials, then $c(x)^{rc}\in C$ for any $c(x)\in C.$
\end{theorem}

\begin{proof}
Since $C=\left\langle g(x),wa(x)\right\rangle ,$ for any $c(x)\in C$, there
exist $m(x)$ and $n(x)$ in $R[x]$ such that $c(x)=g(x)m(x)+wa(x)n(x).$ By
using Lemma 4, we have 
\begin{eqnarray*}
c^{\ast }(x) &=&(g(x)m(x)+wa(x)n(x)) \\
&=&(g(x)m(x))^{\ast }+x^{s}(wa(x)n(x)) \\
&=&g^{\ast }(x)m^{\ast }(x)+wa^{\ast }(x)(x^{s}n^{\ast }(x))
\end{eqnarray*}%
Since $g^{\ast }(x)=eg(x),a^{\ast }(x)=da(x)$, we have $c^{\ast
}(x)=eg(x)m^{\ast }(x)+dwa(x)(x^{s}n^{\ast }(x))\in C$. So, $c^{\ast }(x)\in
C.$

Let $c(x)=c_{0}+c_{1}x+...+c_{t}x^{t}\in C$. Since $C$ is a cyclic code, we
get 
\begin{equation*}
x^{n-t-1}c(x)=c_{0}x^{n-t-1}+c_{1}x^{n-t}+...+c_{t}x^{n-1}\in C
\end{equation*}%
Since $(1+w)+(1+w)x+...+(1+w)x^{n-1}\in C$ and $C$ is a linear code%
\begin{eqnarray*}
-(1+w)\frac{x^{n}-1}{x-1}-x^{n-t-1}c(x)
&=&-(1+w)-(1+w)x+...+(-c_{0}-(1+w))x^{n-t-1} \\
+...+(-c_{t}-(1+w))x^{n-1} &\in &C
\end{eqnarray*}%
By using $\overline{a}+(1+w)=-a$, this implies that 
\begin{equation*}
-(1+w)-...+\overline{c}_{0}x^{n-t-1}+...+\overline{c}_{t}x^{n-1}\in C
\end{equation*}%
This shows that $(c^{\ast }(x))^{rc}\in C$. 
\begin{equation*}
((c^{\ast }(x))^{rc})^{\ast }=\overline{c}_{t}+\overline{c}%
_{t-1}x+...+(3+3w)x^{n-1}
\end{equation*}%
This corresponds this vector $(\overline{c}_{t},\overline{c}_{t-1},...,%
\overline{c}_{0},...,\overline{0})$. Since $(c^{\ast }(x)^{rc})^{\ast
}=(x^{n-t-1}c(x))^{rc}$, so $c(x)^{rc}\in C.$
\end{proof}

\section{The reversible and reversible complement codes over $S$}

Throughout this paper, $S$ denotes the commutative ring $%
Z_{4}+wZ_{4}+vZ_{4}+wvZ_{4}=\{b_{1}+wb_{2}+vb_{3}+wvb_{4}:b_{j}\in
Z_{4},1\leq j\leq 4\}$ with $w^{2}=2,v^{2}=v,wv=vw,$ with characteristic 4. $%
S$ can also be thought of as the quotient ring $%
Z_{4}[w,v]/<w^{2}-2,v^{2}-v,wv-vw>.$

Let 
\begin{eqnarray*}
S &=&Z_{4}+wZ_{4}+vZ_{4}+wvZ_{4},\text{ where }w^{2}=2,v^{2}=v,wv=vw \\
&=&(Z_{4}+wZ_{4})+v(Z_{4}+wZ_{4}),\text{ where }w^{2}=2,v^{2}=v,wv=vw \\
&=&R+vR,\text{ where }v^{2}=v
\end{eqnarray*}

We define the Gray map $\phi _{1}$ from $S$ to $R$ as follows%
\begin{eqnarray*}
\phi _{1} &:&S\longrightarrow R^{2} \\
a+vb &\longmapsto &(a,b)
\end{eqnarray*}%
where $a,b\in R.$ This Gray map is extended compenentwise to 
\begin{eqnarray*}
\phi _{1} &:&S^{n}\longrightarrow R^{2n} \\
x &=&(x_{1},...,x_{n})\longmapsto (a_{1},...,a_{n},b_{1},...,b_{n})
\end{eqnarray*}%
where $x_{i}=a_{i}+vb_{i},a_{i},b_{i}\in R$ for $i=1,2,...,n.$

In this section, we study cyclic codes of odd length $n$ over $S$ satisfy
reverse and reverse complement constraint. Since the ring $S$ is of the
cardinality $4^{4}$, then we define the map $\phi _{1}$ which gives a one to
one correspondence between the element of $S$ and the 256 codons over the
alphabet $\{A,T,G,C\}^{4}$ by using the Gray map. For example;%
\begin{eqnarray*}
0 &=&0+v0\longmapsto \phi _{1}(0)=(0,0)\longrightarrow AAAA \\
2wv &=&0+v(2w)\longmapsto \phi _{1}(2wv)=(0,2w)\longrightarrow AAAG \\
1+3v+3wv &=&1+v(3+3w)\longmapsto \phi
_{1}(1+v(3+3w))=(1,3+3w)\longrightarrow CATT
\end{eqnarray*}

\begin{definition}
Let $A_{1},A_{2}$ be linear codes.%
\begin{equation*}
A_{1}\otimes A_{2}=\{(a_{1},a_{2}):a_{1}\in A_{1},a_{2}\in A_{2}\}
\end{equation*}%
and%
\begin{equation*}
A_{1}\oplus A_{2}=\{a_{1}+a_{2}:a_{1}\in A_{1},a_{2}\in A_{2}\}
\end{equation*}

Let $C$ be a linear code of length $n$ over $S$. Define 
\begin{eqnarray*}
C_{1} &=&\{a:\exists \text{ }b\in R^{n},a+vb\in C\} \\
C_{2} &=&\{b:\exists \text{ }a\in R^{n},a+vb\in C\}
\end{eqnarray*}%
where $C_{1}$ and $C_{2}$ are \ linear codes over $R$ of length $n$.
\end{definition}

\begin{theorem}
Let $C$ be a linear code of length $n$ over $S$. Then $\phi
_{1}(C)=C_{1}\otimes C_{2}$ and $\left\vert C\right\vert =\left\vert
C_{1}\right\vert \left\vert C_{2}\right\vert .$
\end{theorem}

\begin{corollary}
If $\phi _{1}(C)=C_{1}\otimes C_{2}$, then $C=vC_{1}\oplus (1-v)C_{2}.$
\end{corollary}

\begin{theorem}
Let $C=vC_{1}\oplus (1-v)C_{2}$ be a linear code of odd length $n$ over $S$.
Then $C$ is a cyclic code over $S$ if and only if $C_{1},C_{2}$ are cyclic
codes over $R$.
\end{theorem}

\begin{proof}
Let $(a_{0}^{1},a_{1}^{1},...,a_{n-1}^{1})\in
C_{1},(a_{0}^{2},a_{1}^{2},...,a_{n-1}^{2})\in C_{2}$. Assume that $%
m_{i}=va_{i}^{1}\oplus (1-v)a_{i}^{2}$ for $i=0,1,2,...,n-1$. Then $%
(m_{0},m_{1},...,m_{n-1})\in C$. Since $C$ is a cyclic code, it follows that 
$(m_{n-1},m_{0},m_{1},...,m_{n-2})\in C$. Note that $%
(m_{n-1},m_{0},...,m_{n-2})=v(a_{n-1}^{1},a_{0}^{1},...,a_{n-2}^{1})\oplus
(1-v)(a_{n-1}^{2},a_{0}^{2},...,a_{n-2}^{2})$. Hence $%
(a_{n-1}^{1},a_{0}^{1},...,a_{n-2}^{1})\in
C_{1},(a_{n-1}^{2},a_{0}^{2},...,a_{n-2}^{2})\in C_{2}$. Therefore $%
C_{1},C_{2}$ are cyclic codes over $R.$

Conversely, suppose that $C_{1},C_{2}$ are cyclic codes over $R$. Let $%
(m_{0},m_{1},...,m_{n-1})\in C$, where $m_{i}=va_{i}^{1}\oplus
(1-v)a_{i}^{2} $ for $i=0,1,2,...,n-1$. Then $%
(a_{n-1}^{1},a_{0}^{1},...,a_{n-2}^{1})\in
C_{1},(a_{n-1}^{2},a_{0}^{2},...,a_{n-2}^{2})\in C_{2}$. Note that $%
(m_{n-1},m_{0},...,m_{n-2})=v(a_{n-1}^{1},a_{0}^{1},...,a_{n-2}^{1})\oplus
(1-v)(a_{n-1}^{2},a_{0}^{2},...,a_{n-2}^{2})\in C$. So, $C$ is a cyclic code
over $S$.
\end{proof}

\begin{theorem}
Let $C=vC_{1}\oplus (1-v)C_{2}$ be a linear code of odd length $n$ over $S$.
Then $C$ is reversible over $S$ iff $C_{1},C_{2}$ are reversible over $R$.
\end{theorem}

\begin{proof}
Let $C_{1},C_{2}$ be reversible codes. For any $b\in C,b=vb_{1}+(1-v)b_{2},$
where $b_{1}\in C_{1},b_{2}\in C_{2}$. Since $C_{1}$ and $C_{2}$ are
reversible, $b_{1}^{r}\in C_{1},b_{2}^{r}\in C_{2}$. So, $%
b^{r}=vb_{1}^{r}+(1-v)b_{2}^{r}\in C$. Hence $C$ is reversible.

On the other hand, Let $C$ be a reversible code over $S$. So for any $%
b=vb_{1}+(1-v)b_{2}\in C$, where $b_{1}\in C_{1},b_{2}\in C_{2}$, we get $%
b^{r}=vb_{1}^{r}+(1-v)b_{2}^{r}\in C$. Let $%
b^{r}=vb_{1}^{r}+(1-v)b_{2}^{r}=vs_{1}+(1-v)s_{2}$, where $s_{1}\in
C_{1},s_{2}\in C_{2}$. So $C_{1}$ and $C_{2}$ are reversible codes over $R$.
\end{proof}

\begin{lemma}
For any $c\in S$, we have $c+\overline{c}=(3+3w)+v(3+3w).$
\end{lemma}

\begin{lemma}
For any $a\in S$, $\overline{a}+3\overline{0}=3a.$
\end{lemma}

\begin{theorem}
Let $C=vC_{1}\oplus (1-v)C_{2}$ be a cyclic code of odd length $n$ over $S$.
Then $C$ is reversible complement over $S$ iff $C$ is reversible over $S$
and $(\overline{0},\overline{0},...,\overline{0})\in C.$
\end{theorem}

\begin{proof}
Since $C$ is reversible complement, for any $c=(c_{0},c_{1},...,c_{n-1})\in
C,c^{rc}=(\overline{c}_{n-1},\overline{c}_{n-2},...,\overline{c}_{0})\in C$.
Since $C$ is a linear code, so $(0,0,...,0)\in C$. Since $C$ is reversible
complement, so $(\overline{0},\overline{0},...,\overline{0})\in C$. By using
Lemma 15, we have 
\begin{equation*}
3c^{r}=3(c_{n-1},c_{n-2},...,c_{0})=(\overline{c}_{n-1},\overline{c}%
_{n-2},...,\overline{c}_{0})+3(\overline{0},\overline{0},...,\overline{0}%
)\in C
\end{equation*}%
So, for any $c\in C$, we have $c^{r}\in C.$

On the other hand, let $C$ be reversible. So, for any $%
c=(c_{0},c_{1},...,c_{n-1})\in C,c^{r}=(c_{n-1},c_{n-2},...,c_{0})\in C$. To
show that $C$ is reversible complement, for any $c\in C,$%
\begin{equation*}
c^{rc}=(\overline{c}_{n-1},\overline{c}_{n-2},...,\overline{c}%
_{0})=3(c_{n-1},c_{n-2},...,c_{0})+(\overline{0},\overline{0},...,\overline{0%
})\in C
\end{equation*}%
So, $C$ is reversible complement.
\end{proof}

\begin{lemma}
For any $a,b\in S$, $\overline{a+b}=\overline{a}+\overline{b}-3(1+w)(1+v).$
\end{lemma}

\begin{theorem}
Let $D_{1}$ and $D_{2}$ be two reversible complement cyclic codes of length $%
n$ over $S$. Then $D_{1}+D_{2}$ and $D_{1}\cap D_{2}$ are reversible
complement cyclic codes.
\end{theorem}

\begin{proof}
Let $d_{1}=(c_{0},c_{1},...,c_{n-1})\in
D_{1},d_{2}=(c_{0}^{1},c_{1}^{1},...,c_{n-1}^{1})\in D_{2}.$ Then,

\begin{eqnarray*}
(d_{1}+d_{2})^{rc} &=&\left( \overline{(c_{n-1}+c_{n-1}^{1})},...,\overline{%
(c_{1}+c_{1}^{1})},\overline{(c_{0}+c_{0}^{1})}\right) \\
&=&\left( \overline{c_{n-1}}+\overline{c_{n-1}^{1}}-3(1+w)(1+v),...,%
\overline{c_{0}}+\overline{c_{0}^{1}}-3(1+w)(1+v)\right) \\
&=&\left( \overline{c_{n-1}}-3(1+w)(1+v),...,\overline{c_{0}}%
-3(1+w)(1+v)\right) +\left( \overline{c_{n-1}^{1}},...,\overline{c_{0}^{1}}%
\right) \\
&=&\left( d_{1}^{rc}-3(1+w)(1+v)\frac{x^{n}-1}{x-1}\right) +d_{2}^{rc}\in
D_{1}+D_{2}
\end{eqnarray*}%
This shows that $D_{1}+D_{2}$ is reversible complement cyclic code. It is
clear that $D_{1}\cap D_{2}$ is reversible complement cyclic code.
\end{proof}

\section{Binary images of cyclic DNA codes over $R$}

The 2-adic expansion of $c\in Z_{4}$ is $c=\alpha (c)+2\beta (c)$ such that $%
\alpha (c)+\beta (c)+\gamma (c)=0$ for all $c\in Z_{4}$%
\begin{equation*}
\begin{tabular}{cccc}
$c$ & $\alpha (c)$ & $\beta (c)$ & $\gamma (c)$ \\ 
$0$ & $0$ & $0$ & $0$ \\ 
$1$ & $1$ & $0$ & $1$ \\ 
$2$ & $0$ & $1$ & $1$ \\ 
$3$ & $1$ & $1$ & $0$%
\end{tabular}%
\end{equation*}

The Gray map is given by 
\begin{eqnarray*}
\Psi &:&Z_{4}\longrightarrow Z_{2}^{2} \\
c &\longmapsto &\Psi (c)=\left( \beta (c),\gamma (c)\right)
\end{eqnarray*}%
for all $c\in Z_{4}$ in [14]. Define%
\begin{eqnarray*}
\breve{O} &:&R\longrightarrow Z_{2}^{4} \\
a+bw &\longmapsto &\breve{O}(a+wb)=\Psi \left( \phi \left( a+wb\right)
\right) \\
&=&\Psi (a,b) \\
&=&\left( \beta (a),\gamma (a),\beta (b),\gamma (b)\right)
\end{eqnarray*}

Let $a+wb$ be any element of the ring $R$. The Lee weight $w_{L}$ of the
ring $R$ is defined as follows%
\begin{equation*}
w_{L}(a+wb)=w_{L}(a,b)
\end{equation*}%
where $w_{L}(a,b)$ described the usual Lee weight on $Z_{4}^{2}$. For any $%
c_{1},c_{2}\in R$ the Lee distance $d_{L}$ is given by $%
d_{L}(c_{1},c_{2})=w_{L}(c_{1}-c_{2}).$

The Hamming distance $d(c_{1},c_{2})$ between two codewords $c_{1}$ and $%
c_{2}$ is the Hamming weight of the codewords $c_{1}-c_{2}.$%
\begin{equation*}
\begin{tabular}{ccc}
$AA$ & $\longrightarrow $ & $0000$ \\ 
$CA$ & $\longrightarrow $ & $0100$ \\ 
$GA$ & $\longrightarrow $ & $1100$ \\ 
$TA$ & $\longrightarrow $ & $1000$ \\ 
$AC$ & $\longrightarrow $ & $0001$ \\ 
$AG$ & $\longrightarrow $ & $0011$ \\ 
$AT$ & $\longrightarrow $ & $0010$ \\ 
$CC$ & $\longrightarrow $ & $0101$%
\end{tabular}%
\begin{tabular}{ccc}
$CG$ & $\longrightarrow $ & $0111$ \\ 
$CT$ & $\longrightarrow $ & $0110$ \\ 
$GC$ & $\longrightarrow $ & $1101$ \\ 
$GG$ & $\longrightarrow $ & $1111$ \\ 
$GT$ & $\longrightarrow $ & $1110$ \\ 
$TC$ & $\longrightarrow $ & $1001$ \\ 
$TG$ & $\longrightarrow $ & $1011$ \\ 
$TT$ & $\longrightarrow $ & $1010$%
\end{tabular}%
\end{equation*}

\begin{lemma}
The Gray map $\breve{O}$ is a distance preserving map from ($R^{n}$, Lee
distance) to ($Z_{2}^{4n}$, Hamming distance). It is also $Z_{2}$-linear.
\end{lemma}

\begin{proof}
For $c_{1},c_{2}\in R^{n}$, we have $\breve{O}(c_{1}-c_{2})=\breve{O}(c_{1})-%
\breve{O}(c_{2})$. So, $d_{L}(c_{1},c_{2})=w_{L}(c_{1}-c_{2})=w_{H}(\breve{O}%
(c_{1}-c_{2}))=w_{H}(\breve{O}(c_{1})-\breve{O}(c_{2}))=d_{H}(\breve{O}%
(c_{1}),\breve{O}(c_{2}))$. So, the Gray map $\breve{O}$ is distance
preserving map. For any $c_{1},c_{2}\in R^{n},k_{1},k_{2}\in Z_{2},$we have $%
\breve{O}(k_{1}c_{1}+k_{2}c_{2})=k_{1}\breve{O}(c_{1})+k_{2}\breve{O}(c_{2})$%
. Thus, $\breve{O}$ is $Z_{2}$-linear.
\end{proof}

\begin{proposition}
Let $\sigma $ be the cyclic shift of $R^{n}$ and $\upsilon $ be the
4-quasi-cyclic shift of $Z_{2}^{4n}$. Let $\breve{O}$ be the Gray map from $%
R^{n}$ to $Z_{2}^{4n}$. Then $\breve{O}\sigma =\upsilon \breve{O}.$
\end{proposition}

\begin{proof}
Let $c=(c_{0},c_{1},...,c_{n-1})\in R^{n}$, we have $c_{i}=a_{1i}+wb_{2i}$
with $a_{1i},b_{2i}\in Z_{4},0\leq i\leq n-1$. By applying the Gray map, we
have%
\begin{equation*}
\breve{O}(c)=\left( 
\begin{array}{c}
\beta (a_{10}),\gamma (a_{10}),\beta (b_{20}),\gamma (b_{20}),\beta
(a_{11}),\gamma (a_{11}),\beta (b_{21}),\gamma (b_{21}),..., \\ 
\beta (a_{1n-1}),\gamma (a_{1n-1}),\beta (b_{2n-1}),\gamma (b_{2n-1})%
\end{array}%
\right) .
\end{equation*}%
Hence 
\begin{equation*}
\upsilon (\breve{O}(c))=\left( 
\begin{array}{c}
\beta (a_{1n-1}),\gamma (a_{1n-1}),\beta (b_{2n-1}),\gamma (b_{2n-1}),\beta
(a_{10}),\gamma (a_{10}),\beta (b_{20}), \\ 
\gamma (b_{20}),...,\beta (a_{1n-2}),\gamma (a_{1n-2}),\beta
(b_{2n-2}),\gamma (b_{2n-2})%
\end{array}%
\right) .
\end{equation*}

On the other hand, $\sigma (c)=(c_{n-1},c_{0},c_{1},...,c_{n-2})$. We have 
\begin{equation*}
\breve{O}(\sigma (c))=\left( 
\begin{array}{c}
\beta (a_{1n-1}),\gamma (a_{1n-1}),\beta (b_{2n-1}),\gamma (b_{2n-1}),\beta
(a_{10}),\gamma (a_{10}), \\ 
\beta (b_{20}),\gamma (b_{20}),...,\beta (a_{1n-2}),\gamma (a_{1n-2}),\beta
(b_{2n-2}),\gamma (b_{2n-2})%
\end{array}%
\right) .
\end{equation*}%
\newline
Therefore, $\breve{O}\sigma =\upsilon \breve{O}.$
\end{proof}

\begin{theorem}
If $C$ is a cyclic DNA code of length $n$ over $R$ then $\breve{O}$ is a
binary quasi-cyclic DNA code of length $4n$ with index $4$.
\end{theorem}

\section{Binary image of cyclic DNA codes over $S$}

We define%
\begin{eqnarray*}
\widetilde{\Psi } &:&S\longrightarrow Z_{4}^{4} \\
a_{0}+wa_{1}+va_{2}+wva_{3} &\longmapsto &(a_{0},a_{1},a_{2},a_{3})
\end{eqnarray*}%
where $a_{i}\in Z_{4}$ for $i=0,1,2,3$.

Now, we define%
\begin{eqnarray*}
\Theta  &:&S\longrightarrow Z_{2}^{8} \\
a_{0}+wa_{1}+va_{2}+wva_{3} &\longmapsto &\Theta
(a_{0}+wa_{1}+va_{2}+wva_{3})=\Psi (\widetilde{\Psi }%
(a_{0}+wa_{1}+va_{2}+wva_{3})) \\
&=&\left( \beta (a_{0}),\gamma (a_{0}),\beta (a_{1}),\gamma (a_{1}),\beta
(a_{2}),\gamma (a_{2}),\beta (a_{3}),\gamma (a_{3})\right) 
\end{eqnarray*}%
where $\Psi $ is the Gray map $Z_{4}$ to $Z_{2}^{2}.$

Let $a_{0}+wa_{1}+va_{2}+wva_{3}$ be any element of the ring $S$. The Lee
weight $w_{L}$ of the ring $S$ is defined as%
\begin{equation*}
w_{L}(a_{0}+wa_{1}+va_{2}+wva_{3})=w_{L}((a_{0},a_{1},a_{2},a_{3}))
\end{equation*}%
where $w_{L}((a_{0},a_{1},a_{2},a_{3}))$ described the usual Lee weight on $%
Z_{4}^{4}$. For any $c_{1},c_{2}\in S,$ the Lee distance $d_{L}$ is given by 
$d_{L}(c_{1},c_{2})=w_{L}(c_{1}-c_{2}).$

The Hamming distance $d(c_{1},c_{2})$ between two codewords $c_{1}$ and $%
c_{2}$ is the Hamming weight of the codewords $c_{1}-c_{2}.$

Binary image of the codons;%
\begin{equation*}
\begin{tabular}{ccc}
$AAAA$ & $\longrightarrow $ & $00000000$ \\ 
$AACA$ & $\longrightarrow $ & $00000100$ \\ 
$AAGA$ & $\longrightarrow $ & $00001100$ \\ 
$AATA$ & $\longrightarrow $ & $00001000$ \\ 
$\vdots $ & $\vdots $ & $\vdots $%
\end{tabular}%
\end{equation*}

\begin{lemma}
The Gray map $\Theta $ is a distance preserving map from ($S^{n}$, Lee
distance) to ($Z_{2}^{8n}$, Hamming distance). It is also $Z_{2}$-linear.
\end{lemma}

\begin{proof}
It is proved as in the proof of Lemma 19.
\end{proof}

\begin{proposition}
Let $\sigma $ be the cyclic shift of $S^{n}$ and $\overset{\prime }{\upsilon 
}$ be the 8-quasi-cyclic shift of $Z_{2}^{8n}$. Let $\Theta $ be the Gray
map from $S^{n}$ to $Z_{2}^{8n}$. Then $\Theta \sigma =\overset{\prime }{%
\upsilon }\Theta .$
\end{proposition}

\begin{proof}
It is proved as in the proof of Proposition 20.
\end{proof}

\begin{theorem}
If $C$ is a cyclic DNA code of length $n$ over $S$ then $\Theta $ is a
binary quasi-cyclic DNA code of length $8n$ with index $8$.
\end{theorem}

\begin{proof}
Let $C$ be a cyclic DNA code of length $n$ over $S$. So, $\sigma (C)=C$. By
using the Proposition 23, we have $\Theta (\sigma (C))=\overset{\prime }{%
\upsilon }(\Theta (C))=\Theta (C)$. Hence $\Theta (C)$ is a set of length $8n
$ over the alphabet $Z_{2}$ which is a quasi-cyclic code of index 8.
\end{proof}

\section{Skew cyclic DNA\ codes over $R$}

In [6], the skew codes over $R$ were studied and the Gray images of them
were determined.

We will use the non trivial automorphism in [6]. For all $a+wb\in R$, it was
defined by 
\begin{eqnarray*}
\theta &:&R\longrightarrow R \\
a+wb &\longmapsto &a-wb
\end{eqnarray*}

The ring $R[x,\theta ]=\{a_{0}+a_{1}x+...+a_{n-1}x^{n-1}:a_{i}\in R,n\in N\}$
is called skew polynomial ring. It is non commutative ring. The addition in
the ring $R[x,\theta ]$ is the usual polynomial and multiplication is
defined as $(ax^{i})(bx^{j})=a\theta ^{i}(b)x^{i+j}.$ The order of the
automorphism $\theta $ is 2.

The following a definition and three theorems are in [6].

\begin{definition}
A subset $C$ of $R^{n}$ is called a skew cyclic code of length $n$ if $C$
satisfies the following conditions,

$i)$ $C$ is a submodule of $R^{n}$,

$ii)$ If $c=\left( c_{0},c_{1},...,c_{n-1}\right) \in C$, then $\sigma
_{\theta }\left( c\right) =\left( \theta (c_{n-1}),\theta (c_{0}),...,\theta
(c_{n-2})\right) \in C$
\end{definition}

Let $f(x)+\left\langle x^{n}-1\right\rangle $ be an element in the set $\ 
\check{R}_{n}=R\left[ x,\theta \right] /\left\langle x^{n}-1\right\rangle $
and let $r(x)\in R\left[ x,\theta \right] $. Define multiplication from left
as follows,

\begin{equation*}
r(x)(f(x)+\left\langle x^{n}-1\right\rangle )=r(x)f(x)+\left\langle
x^{n}-1\right\rangle
\end{equation*}%
for any $r(x)\in R\left[ x,\theta \right] $.

\begin{theorem}
$\check{R}_{n}$ is a left $R\left[ x,\theta \right] $-module where
multiplication defined as in above.
\end{theorem}

\begin{theorem}
A code $C$ over $R$ of length $n$ is a skew cyclic code if and only if $C$
is a left $R\left[ x,\theta \right] $-submodule of the left $R\left[
x,\theta \right] $-module $\check{R}_{n}$.
\end{theorem}

\begin{theorem}
Let $C$ be a skew cyclic code over $R$ of length $n$ and let $f(x)$ be a
polynomial in $C$ of minimal degree. If $f(x)$ is monic polynomial, then $%
C=\left\langle f(x)\right\rangle ,$ where $f(x)$ is a right divisor of $%
x^{n}-1.$
\end{theorem}

For all $x\in R$, we have%
\begin{equation*}
\theta (x)+\theta (\overline{x})=3-3w
\end{equation*}

\begin{theorem}
Let $C=\left\langle f(x)\right\rangle $ be a skew cyclic code over $R,$
where $f(x)$ is a monic polynomial in $C$ of minimal degree. If $C$ is
reversible complement, the polynomial $f(x)$ is self reciprocal and $(3+3w)%
\frac{x^{n}-1}{x-1}\in C.$
\end{theorem}

\begin{proof}
Let $C=\left\langle f(x)\right\rangle $ be a skew cyclic code over $R,$
where $f(x)$ is a monic polynomial in $C$. Since $\left( 0,0,...,0\right)
\in C$ and $C$ is reversible complement, we have $\left( \overline{0},%
\overline{0},...,\overline{0}\right) =(3+3w,3+3w,...,3+3w)\in C.$

Let $f(x)=1+a_{1}x+...+a_{t-1}x^{t-1}+x^{t}$. Since $C$ is reversible
complement, we have $f^{rc}(x)\in C$. That is 
\begin{eqnarray*}
f^{rc}(x) &=&(3+3w)+(3+3w)x+...+(3+3w)x^{n-t-2}+(2+3w)x^{n-t-1}+ \\
&&\overline{a}_{t-1}x^{n-t}+...+\overline{a}_{1}x^{n-2}+(2+3w)x^{n-1}
\end{eqnarray*}

Since $C$ is a linear code, we have $f^{rc}(x)-(3+3w)\frac{x^{n}-1}{x-1}\in
C $. This implies that 
\begin{equation*}
-x^{n-t-1}+(\overline{a}_{t-1}-(3+3w))x^{n-t}+...+(\overline{a}%
_{1}-(3+3w))x^{n-2}-x^{n-1}\in C
\end{equation*}%
multiplying on the right by $x^{t+1-n}$, we have 
\begin{equation*}
-1+(\overline{a}_{t-1}-(3+3w))\theta (1)x+...+(\overline{a}%
_{1}-(3+3w))\theta ^{t-1}(1)x^{t-1}-\theta ^{t}(1)x^{t}\in C
\end{equation*}%
By using $a+\overline{a}=3+3w$, we have%
\begin{equation*}
-1-a_{t-1}x-a_{t-2}x^{2}-...-a_{1}x^{t-1}-x^{t}=3f^{\ast }(x)\in C
\end{equation*}

Since $C=\left\langle f(x)\right\rangle $, there exist $q(x)\in R\left[
x,\theta \right] $ such that $3f^{\ast }(x)=q(x)f(x)$. Since $\deg f(x)=\deg
f^{\ast }(x),$ we have $q(x)=1$. Since $3f^{\ast }(x)=f(x)$, we have $%
f^{\ast }(x)=3f(x)$. So, $f(x)$ is self reciprocal.
\end{proof}

\begin{theorem}
Let $C=\left\langle f(x)\right\rangle $ be a skew cyclic code over $R$,
where $f(x)$ is a monic polynomial in $C$ of minimal degree. If $(3+3w)\frac{%
x^{n}-1}{x-1}\in C$ and $f(x)$ is self reciprocal, then $C$ is reversible
complement.
\end{theorem}

\begin{proof}
Let $f(x)=1+a_{1}x+...+a_{t-1}x^{t-1}+x^{t}$ be a monic polynomial of the
minimal degree.

Let $c(x)\in C$. So, $c(x)=q(x)f(x)$, where $q(x)\in R[x,\theta ]$. By using
Lemma 4, we have $c^{\ast }(x)=(q(x)f(x))^{\ast }=q^{\ast }(x)f^{\ast }(x).$
Since $f(x)$ is self reciprocal, so $c^{\ast }(x)=q^{\ast }(x)ef(x)$, where $%
e\in Z_{4}\backslash \{0\}$. Therefore $c^{\ast }(x)\in C=\left\langle
f(x)\right\rangle $. Let $c(x)=c_{0}+c_{1}x+...+c_{t}x^{t}\in C.$ Since $C$
is a cyclic code, we get%
\begin{equation*}
c(x)x^{n-t-1}=c_{0}x^{n-t-1}+c_{1}x^{n-t}+...+c_{t}x^{n-1}\in C
\end{equation*}%
The vector correspond to this polynomial is 
\begin{equation*}
(0,0,...,0,c_{0},c_{1},...,c_{t})\in C
\end{equation*}%
Since $(3+3w,3+3w,...,3+3w)\in C$ and $C$ linear, we have%
\begin{eqnarray*}
(3+3w,3+3w,...,3+3w)-(0,0,...,0,c_{0},c_{1},...,c_{t}) &=&(3+3w,...,3+3w, \\
(3+3w)-c_{0},...,(3+3w)-c_{t}) &\in &C
\end{eqnarray*}%
By using $a+\overline{a}=3+3w$, we get 
\begin{equation*}
(3+3w,3+3w,...,3+3w,\overline{c}_{0},...,\overline{c}_{t})\in C
\end{equation*}%
which is equal to $\left( c(x)^{\ast }\right) ^{rc}$. This shows that $%
\left( c(x)^{\ast }\right) ^{rc}=c(x)^{rc}\in C.$
\end{proof}

\section{DNA codes over $S$}

\begin{definition}
Let $f_{1}$ and $f_{2}$ be polynomials with $\deg f_{1}=t_{1},\deg
f_{2}=t_{2}$ and both dividing $x^{n}-1\in R$

Let $m=\min \{n-t_{1},n-t_{2}\}$ and $f(x)=vf_{1}(x)+(1-v)f_{2}(x)$ over $S$%
. The set $L(f)$ is called a $\Gamma $-set where the automorphism $\Gamma $
is defined as follows;%
\begin{eqnarray*}
\Gamma &:&S\longrightarrow S \\
a+wb+vc+wvd &\longmapsto &a+b+w(b+d)-vc-wvdc
\end{eqnarray*}

The set $L(f)$ is defined as $L(f)=\{E_{0},E_{1},...,E_{m-1}\}$, where%
\begin{equation*}
E_{i}=\left\{ 
\begin{array}{c}
x^{i}f\text{ \ \ if }i\text{ is even} \\ 
x^{i}\Gamma (f)\text{ if }i\text{ is odd}%
\end{array}%
\right.
\end{equation*}

$L(f)$ generates a linear code $C$ over $S$ denoted by $C=\left\langle
f\right\rangle _{\Gamma }$. Let $f(x)=a_{0}+a_{1}x+...+a_{t}x^{t}$ be over $%
S $ and $S$-submodule generated by $L(f)$ is generated by the following
matrix%
\begin{equation*}
L(f)=\left[ 
\begin{tabular}{llllllllll}
$a_{0}$ & $a_{1}$ & $a_{2}$ & $\cdots $ & $a_{t}$ & $0$ & $\cdots $ & $%
\cdots $ & $\cdots $ & $0$ \\ 
$0$ & $\Gamma (a_{0})$ & $\Gamma (a_{1})$ & $\cdots $ & $\cdots $ & $\Gamma
(a_{t})$ & $0$ & $\cdots $ & $\cdots $ & $0$ \\ 
$0$ & $0$ & $a_{0}$ & $a_{1}$ & $\cdots $ & $\cdots $ & $a_{t}$ & $0$ & $%
\cdots $ & $0$ \\ 
$0$ & $0$ & $0$ & $\Gamma (a_{0})$ & $\Gamma (a_{1})$ & $\cdots $ & $\cdots $
& $\Gamma (a_{t})$ & $\cdots $ & $0$ \\ 
$\vdots $ & $\cdots $ & $\cdots $ & $\cdots $ & $\vdots $ & $\cdots $ & $%
\cdots $ & $\cdots $ & $\cdots $ & $\vdots $%
\end{tabular}%
\right]
\end{equation*}
\end{definition}

\begin{theorem}
Let $f_{1}$ and $f_{2}$ be self reciprocal polynomials dividing $x^{n}-1$
over $R$ with degree $t_{1}$ and $t_{2}$, respectively. If $f_{1}=f_{2},$
then $f=vf_{1}+(1-v)f_{2}$ and $\left\vert \left\langle L(f)\right\rangle
\right\vert =256^{m}.$ $C=\left\langle L(f)\right\rangle $ is a linear code
over $S$ and $\Theta (C)$ is a reversible DNA code.
\end{theorem}

\begin{proof}
It is proved as in the proof of the Theorem 5 in [4].
\end{proof}

\begin{corollary}
Let $f_{1}$ and $f_{2}$ be self reciprocal polynomials dividing $x^{n}-1$
over $R$ and $C=\left\langle L(f)\right\rangle $ be a cyclic code over $S.$
If $\frac{x^{n}-1}{x-1}\in C$, then $\Theta (C)$ is a reversible complement
DNA code.
\end{corollary}

\begin{example}
Let $f_{1}(x)=f_{2}(x)=x-1$ dividing $x^{7}-1$ over $R.$ Hence, $%
C=\left\langle vf_{1}(x)+(1-v)f_{2}(x)\right\rangle _{\Gamma }=\left\langle
x-1\right\rangle _{\Gamma }$ is a $\Gamma $-linear code over $S$ and $\Theta
(C)$ is a reversible complement DNA code, because of $\frac{x^{n}-1}{x-1}\in
C.$
\end{example}

\section{References}

\ \ [1] Abualrub T., Ghrayeb A., Zeng X., Construction of cyclic codes over $%
GF(4)$ for DNA computing, J. Franklin Institute, 343, 448-457, 2006.

[2] Abualrub T., Siap I., Reversible quaternary cyclic codes, Proc. of the
9th WSEAS Int. Conference on Appl. Math., Istanbul, 441-446, 2006.

[3] Adleman L., Molecular computation of the solution to combinatorial
problems, Science, 266, 1021-1024, 1994.

[4] Bayram A., Oztas E., Siap I., Codes over $F_{4}+vF_{4}$ and some DNA
applications, Designs, Codes and Cryptography, DOI:
10.107/s10623-015-0100-8, 2015.

[5] Bennenni N., Guenda K., Mesnager S., New DNA cyclic codes over rings,
arXiv: 1505.06263v1, 2015.

[6] Dertli A., Cengellenmis Y., Eren S., On the codes over the $Z_{4}+wZ_{4}$%
; Self dual codes, Macwilliams identities, Cyclic, constacyclic and
quasi-cyclic codes, their skew codes, Int. J. of Foundations of Computer
Science, to \ be submitted, 2016.

[7] Gaborit P., King O. D., Linear construction for DNA codes, Theor.
Computer Science, 334, 99-113, 2005.

[8] Guenda K., Gulliver T. A., Sole P,. On cyclic DNA codes, Proc., IEEE
Int. Symp. Inform. Theory, Istanbul, 121-125, 2013.

[9] Guenda K., Gulliver T. A., Construction of cyclic codes over $%
F_{2}+uF_{2}$ for DNA computing, AAECC, 24, 445-459, 2013.

[10] Liang J., Wang L., On cyclic DNA codes over $F_{2}+uF_{2},$ J. Appl.
Math. Comput., DOI: 10.1007/s12190-015-0892-8, 2015.

[11] Ma F., Yonglin C., Jian G., On cyclic DNA codes over $%
F_{4}[u]/\left\langle u^{2}+1\right\rangle $.

[12] Massey J. L., Reversible codes, Inf. Control, 7, 369-380, 1964.

[13] Oztas E. S., Siap I., Lifted polynomials over $F_{16}$ and their
applications to DNA codes, Filomat, 27, 459-466, 2013.

[14] Pattanayak S., Singh A. K., On cyclic DNA\ codes over the ring $%
Z_{4}+uZ_{4},$ arXiv: 1508.02015, 2015.

[15] Pattanayak S., Singh A. K., Kumar P., DNA cyclic codes over the ring $%
F_{2}[u,v]/\left\langle u^{2}-1,v^{3}-v,uv-vu\right\rangle ,$
arXiv:1511.03937, 2015.

[16] Pattanayak S., Singh A. K., Construction of Cyclic DNA codes over the
ring $Z_{4}[u]/\left\langle u^{2}-1\right\rangle $ based on deletion
distance, arXiv: 1603.04055v1, 2016.

[17] Siap I., Abualrub T., Ghrayeb A., Cyclic DNA codes over the ring $%
F_{2}[u]/\left( u^{2}-1\right) $ based on the delition distance, J. Franklin
Institute, 346, 731-740, 2009.

[18] Siap I., Abualrub T., Ghrayeb A., Similarity cyclic DNA codes over
rings, IEEE, 978-1-4244-1748-3, 2008.

[19] Y\i ld\i z B., Siap I., Cyclic DNA codes over the ring $F_{2}[u]/\left(
u^{4}-1\right) $ and applications to DNA codes, Comput. Math. Appl., 63,
1169-1176, 2012.

[20] Zhu S., Chen X., Cyclic DNA codes over $F_{2}+uF_{2}+vF_{2}+uvF_{2},$
arXiv: 1508.07113v1, 2015.

\end{document}